\documentclass[12pt]{article}
\usepackage{amsfonts}
\usepackage{amssymb,amsmath,amsthm,latexsym}
\usepackage[numbers,compress]{natbib}
\usepackage{amsmath}
\usepackage{cases}
\usepackage{indentfirst}
\usepackage{mathrsfs}
\usepackage{amsfonts}
\usepackage{amsmath}
\usepackage{amsthm}
\usepackage[usenames]{color}
\usepackage{footnote}
\usepackage{longtable}
\usepackage{hyperref}
\usepackage[all]{xy}
\usepackage{amsmath,amscd}
\usepackage{multirow}
\usepackage{booktabs}
\usepackage{clrscode}
\usepackage{listings}
\usepackage{empheq}
\usepackage{color}
\usepackage{array}

\allowdisplaybreaks[4]
\newtheorem{theorem}{Theorem}[section]
\newtheorem{proposition}[theorem]{Proposition}
\newtheorem{corollary}[theorem]{Corollary}
\newtheorem{lemma}[theorem]{Lemma}

\newtheorem*{remark}{Remark}
\newtheorem{define}[theorem]{Definition}

\numberwithin{equation}{section}

\usepackage{color}
\usepackage{graphicx}
\usepackage{url}
\usepackage{booktabs} 
\usepackage{placeins} 

\def\|{|\kern-1.25pt|}

\def\bfs#1{{\setbox0=\hbox{$\scriptstyle#1$} 
     \kern-.020em\copy0\kern-\wd0
     \kern .040em\copy0\kern-\wd0
     \kern-.020em\raise.01em\box0 }}

\newcommand{\rmv}[1]{}
\usepackage{transparent}
\usepackage{tabularx}
\usepackage{float}

\begin{document}

\title{Low differentially uniform permutations from Dobbertin APN function over $\mathbb{F}_{2^n}$ }
\author{Yan-Ping Wang$^{1,2}$, WeiGuo Zhang$^{1}${\thanks{Email:~zwg@xidian.edu.cn} }, Zhengbang Zha$^3$ \\
\vspace*{0.0cm}\\
{\small $^1$State Key Laboratory of Integrated Service Networks, Xidian University, Xi'an 710071, China}\\
{\small $^2$College of Mathematics and Statistics, Northwest Normal University, Lanzhou 730070, China}\\
{\small $^3$School of Mathematical Sciences, Luoyang Normal University, Luoyang 471934, China}}
\date{}
\maketitle
\begin{abstract}
Block ciphers use S-boxes to create confusion in the cryptosystems.
Such S-boxes are functions over $\mathbb{F}_{2^{n}}$. These functions should have low differential uniformity,
high nonlinearity, and high algebraic degree in order to resist differential attacks, linear attacks, and higher order differential attacks, respectively.
In this paper, we construct new classes of differentially $4$ and $6$-uniform permutations by modifying the image of the Dobbertin APN function $x^{d}$ with $d=2^{4k}+2^{3k}+2^{2k}+2^{k}-1$ over a subfield of $\mathbb{F}_{2^{n}}$. Furthermore, the algebraic degree and the lower bound of the nonlinearity of the constructed functions are given.
\end{abstract}




\section{Introduction}
\label{intro}

Let $\mathbb{F}_{2^n}$ be the finite field with $2^{n}$ elements, and $\mathbb{F}_{2^{n}}^{*}$ be the multiplicative group which consists of all the nonzero elements of $\mathbb{F}_{2^{n}}$. For a function $f:\mathbb{F}_{2^{n}}\rightarrow\mathbb{F}_{2^{n}}$, the derivative of $f(x)$ is defined by $D_{a}f(x) = f(x+a) + f(x)$, where $x\in\mathbb{F}_{2^n}, ~a\in\mathbb{F}_{2^{n}}^{*}$.
For any $b\in\mathbb{F}_{2^{n}}$, we define $\delta_{f}(a,b) = \big|\{x\in\mathbb{F}_{2^{n}} ~|~ D_{a}f(x)=b\} \big|$.
The maximum value $\delta_{f}  = \max\limits_{a\in\mathbb{F}_{2^{n}}^{*},b\in\mathbb{F}_{2^{n}}}\delta_{f}(a,b)$ is called the differential uniformity of $f(x)$.
Let $\omega_{i}(f)=\big|\{(a,b)\in\mathbb{F}^{*}_{2^n}\times\mathbb{F}_{2^n} ~|~ \delta_{f}(a,b)= i \}\big|$. The set $\{\omega_{0}(f), \omega_{2}(f),\cdots, \omega_{\delta_{f}}(f)\}$ is called the differential spectrum of $f(x)$ with the differential uniformity $\delta_{f}$.

Block ciphers should be designed to resist all classical attacks. S-boxes are the core components of block ciphers. The primary purpose of S-boxes is to produce confusion inside block ciphers.
Such S-boxes are nonlinear functions over $\mathbb{F}_{2^{n}}$. Such functions should have low differential uniformity for resisting differential attacks \cite{BS1991}, high nonlinearity for avoiding linear cryptanalysis \cite{ML1994} and high algebraic degree to prevent higher order differential attacks \cite{KL1995}.
It is well known that the lowest differential uniformity achieved by a function $f$ over $\mathbb{F}_{2^n}$ is $\delta_{f}= 2$ and such functions $f$ are called almost perfect nonlinear (APN) functions.
APN functions have important applications in block ciphers. For example, the APN functions $x^{81}$ over $\mathbb{F}_{2^9}$ and $x^5$ over $\mathbb{F}_{2^7}$ have been respectively used in  MISTY and KASUMI block ciphers.
For ease of the implementation in both hardware and software, such functions are required to be defined on $\mathbb{F}_{2^n}$ for even $n$.
It is known that no APN permutations  exist over $\mathbb{F}_{2^n}$ for $n = 2, 4$. Instead, an APN permutation of the field $\mathbb{F}_{2^6}$ was discovered by Dillon et al. in \cite{BDMW2010}.
It is an open problem whether there exists an APN permutation over $\mathbb{F}_{2^n}$ for even $n\geq8$.
So, to resist differential attacks in even dimensions, we can choose differentially $4$-uniform permutations as S-boxes.
A well known example is the multiplicative inverse functions used in the S-box of the Advanced Encryption Standard (AES). This function is differentially $4$-uniform with known maximum nonlinearity and optimal algebraic degree.
A number of recent research works have been devoted to constructing differentially $4$ and $6$-uniform permutation functions with high nonlinearity and high algebraic degree over $\mathbb{F}_{2^{n}}$.

The primary constructions do not need to use previously known functions for designing new functions.
Using the primary constructions, researchers found many classes of differentially $4$-uniform permutation functions on $\mathbb{F}_{2^n}$: Gold functions~\cite{GR1968}, Kasami functions~\cite{KT1971}, Inverse functions~\cite{NK1994}, Bracken-Leander functions~\cite{BG2010} and Bracken-Tan-Tan binomial functions~\cite{BTT2011} and so on. Especially, Nakagawa et al.~\cite{NY2007} constructed a family of quadrinomial differentially $4${\color{red}-}uniform functions by using Albert's finite commutative semifields. However, except Kasami functions and Inverse functions, the algebraic degrees of the other functions
are $2$ or $3$. Note that, to resist the higher order differential attack, we require the algebraic degree to be greater than $3$.

In other ways, the secondary constructions have been proposed to get new differentially $4$ and $6$-uniform permutation functions by modifying known functions.
For example, Li et al.~\cite{LW2013} constructed differentially $4$-uniform permutations on $\mathbb{F}_{2^{2m}}$ via applying APN permutations over $\mathbb{F}_{2^{2m+1}}$.
In recent years, many new constructions consist of piecewise functions, hence functions of the form
\begin{equation} \label{eq:piecewise}
 f(x)= \left\{
                               \begin{array}{ll}
                                g(x), & \textrm{~$x\in \mathbb{S}$},\\
                                h(x), & \textrm{~$x\in \mathbb{F}_{2^n}\backslash \mathbb{S}$}.
                               \end{array}\right.
\end{equation}
The function $f(x)$ is obtained by modifying the image values of $h(x)$ on a subset $\mathbb{S}$ of $\mathbb{F}_{2^n}$. Some new results are obtained from modifying the inverse function $h(x)=x^{-1}$ in (\ref{eq:piecewise}) on a subset of $\mathbb{F}_{2^n}$~\cite{PTW2017,PTW2016,PT2017,QTTL2013,QTLG2016,TCT2015,YWL2011,ZHS2014}. 
More new functions were designed by changing the Gold function $h(x)=x^{2^k+1}$ \cite{CM2019,ZHSJ2014,ZD2012}, the Kasami function $h(x)=x^{2^{2k}-2^{k}+1}$ \cite{XZ2013} or the Bracken-Leander function $h(x)=x^{2^{2k}+2^{k}+1}$ \cite{CM2019} in (\ref{eq:piecewise}) on a subset $\mathbb{S}$ of $\mathbb{F}_{2^n}$.
In this paper, using the piecewise functions method, we present new differentially $4$ and $6$-uniform permutation functions by modifying the Dobbertin APN function $x^{d}$ with $d=2^{4k}+2^{3k}+2^{2k}+2^{k}-1$ (see \cite{Hans2001}) over a subfield of $\mathbb{F}_{2^{n}}$.

The paper is organized as follows. Section \ref{pre} gives some preliminaries on necessary concepts and related results.
In Section \ref{three}, the new differentially $4$ and $6$-uniform permutation functions are presented by modifying the Dobbertin APN function $x^{d}$ with $d=2^{4k}+2^{3k}+2^{2k}+2^{k}-1$ over a subfield of $\mathbb{F}_{2^{n}}$. In Section \ref{conclu}, the conclusion is given.

\section{Preliminaries }\label{pre}

In this section, we give some necessary definitions and results which will be used in this paper.

\begin{define}\label{trace}\emph{(\cite{LN97})}
Let $p$ be a prime and $n$, $m$ be positive integers. Set $q=p^n$. For $\alpha\in\mathbb{F}_{q^m}$, the trace ${\rm Tr}_{\mathbb{F}_{q^m}/{\mathbb{F}_{q}}}(\alpha)$
of $\alpha$ over $\mathbb{F}_{q}$ is defined by
$${\rm Tr}_{\mathbb{F}_{q^m}/{\mathbb{F}_{q}}}(\alpha) = \alpha+ \alpha^{q}+ x^{q^{2}}+ \cdots + \alpha^{q^{m-1}}.$$
Especially, when $n=1$, ${\rm Tr}_{\mathbb{F}_{q^m}/{\mathbb{F}_{p}}}(\alpha)$ is called the absolute trace of $\alpha$ and simply denoted by ${\rm Tr}(\alpha)$, where $\mathbb{F}_{p}$ is the prime subfield of $\mathbb{F}_{q^m}$.
\end{define}

\begin{define}  \emph{(\cite{Carlet2010})} \label{def:walsh}
For a function $f:\mathbb{F}_{2^{n}}\rightarrow \mathbb{F}_{2^{n}}$, the Walsh transform $f^{\mathcal{W}}:\mathbb{F}_{2^{n}}\times\mathbb{F}_{2^{n}}^{*}\rightarrow\mathbb{C}$ of $f$
is defined by \\
$$f^{\mathcal{W}}(u, v)=\sum\limits_{x\in\mathbb{F}_{2^n}}(-1)^{{\rm Tr}(ux+vf(x))},~~~u\in\mathbb{F}_{2^n}, v\in\mathbb{F}^{*}_{2^n}.$$ \\
The set $\mathcal{W}_{f} :=\{f^{\mathcal{W}}(u, v)~|~u\in\mathbb{F}_{2^n}, v\in\mathbb{F}^{*}_{2^n} \}$ is called the Walsh spectrum of $f$. The nonlinearity $\mathcal{NL}(f)$ of $f$ is defined as
$$\mathcal{NL}(f)=2^{n-1}-\frac{1}{2}\max\limits_{w\in\mathcal{W}_{f}} |w|.$$
\end{define}
For odd $n$, it has been shown that $\mathcal{NL}(f)$ is upper bounded by $2^{n-1} - 2^{\frac{n-1}{2}}$.
For even $n$, it is conjectured that $\mathcal{NL}(f)$ is upper bounded by $2^{n-1} - 2^{\frac{n}{2}}$.

\begin{define} \emph{(\cite{Carlet2010})}\label{def:deg}
For $f(x)=\sum\limits_{i=0}^{2^n-1}a_{i}x^{i}$, where $a_{i}\in \mathbb{F}_{2^n}$, we define the algebraic degree of $f(x)$ as $deg(f)=\max\limits_{i=0,\cdots, 2^n-1~ s.t. ~a_{i}\neq0}\omega_{2}(i)$, where $\omega_{2}(i)$ is the $2$-weight of the exponent $i$.
\end{define}

\begin{define} \emph{(\cite{Carlet2010})}
Let the functions $f,~g:~\mathbb{F}_{2^n}\rightarrow \mathbb{F}_{2^n}$. Then $f$ and $g$ are called

\emph{(i)} affine equivalent if $g=\mathcal{L}_{1} \circ f \circ \mathcal{L}_{2} $ for affine permutations  $\mathcal{L}_{1}$ and $\mathcal{L}_{2}$ over $\mathbb{F}_{2^n}$. Furthermore, $f,~g$ are called extended affine equivalent (EA-equivalent) if $g=\mathcal{L}_{1} \circ f \circ \mathcal{L}_{2} +\mathcal{L}_{3}$ for an affine function $\mathcal{L}_{3}$ over $\mathbb{F}_{2^n}$.

\emph{(ii)} Carlet-Charpin-Zinoviev equivalent (CCZ-equivalent) if the graphs of $f$ and $g$ are affine equivalent for affine permutation $\mathcal{A}\in\mathbb{F}_{2^n}\times \mathbb{F}_{2^n}$, that is $\mathcal{A}(G_{f}) =G_{g}$, where $G_{f} = \{(x,f(x))~:~x \in\mathbb{F}_{2^n} \}$ and $G_{g} = \{(x,g(x))~:~x \in\mathbb{F}_{2^n} \}$.
\end{define}

It is easy to see that affine equivalence is a particular case of EA-equivalence. EA-equivalence
implies CCZ-equivalence, and not vice versa. Every permutation is CCZ-equivalent to its inverse, and
the differential spectrum and the Walsh spectrum are CCZ-invariant. When a function is not affine, its algebraic degree is EA-invariant, but not CCZ-invariant. Refer to more details in \cite{Carlet2010}.

To investigate the solutions of a system of polynomial equations, the resultant of two polynomials is needed.

\begin{define}\emph{({\rm \cite[p.36]{LN97}})}
Let $q=p^r$, where $p$ is a prime and $r$ is a positive integer. Let $f(x)=a_{0}x^{n}+a_{1}x^{n-1}+\cdots +a_{n}\in \mathbb{F}_q[x]$ and $g(x)=b_{0}x^{m}+b_{1}x^{m-1}+\cdots+ b_{m}\in \mathbb{F}_q[x]$ be
two polynomials of  degree $n$ and $m$ respectively,  where $n, m \in \mathbb{N}$. Then the resultant $Res(f, g)$ of the two polynomials is defined by the determinant
\begin{eqnarray*}
Res(f, g)=
\left|\begin{array}{cccccccc}
   a_{0} &    a_{1}    & \cdots &   a_{n}  & 0  &    &\cdots & 0   \\
    0  &  a_{0} &    a_{1}    & \cdots &   a_{n}  & 0  & \cdots & 0 \\
    \vdots  &    &        &      &     &      &          &  \vdots    \\
    0  &  \cdots &   0   &  a_{0} &    a_{1}  &     &  \cdots &   a_{n} \\
    b_{0} &    b_{1}    & \cdots &      & b_{m}  & 0  & \cdots &  0     \\
      0   &    b_{0} &   b_{1}  & \cdots &    &  b_{m}  & \cdots &  0 \\
    \vdots &          &      &     &      &     &    &\vdots       \\
    0  &   \cdots   & 0 &  b_{0} &    b_{1} &    & \cdots &  b_{m}  \\
\end{array}\right|
& \begin{array}{l}
\left.\rule{0mm}{9.80mm}\right\}$m$~rows \\
\\\left.\rule{0mm}{9.80mm}\right\}$n$~ rows
\end{array}\\[0pt]
\end{eqnarray*}
of order $m+n$.
\end{define}
If the degree of $f$ is $Deg(f) = n$  (i.e., $a_{0}\neq 0$) and $f(x)=a_{0}(x -\alpha_{1})(x -\alpha_{2})\cdots (x -\alpha_{n})$ in
the splitting field of $f$ over $\mathbb{F}_q$, then $Res(f, g)$ is also given by the formula
\begin{eqnarray*}
Res(f, g)=a_{0}^{m}\prod_{i=1}^n g(\alpha_{i}).
\end{eqnarray*}
In this case, we have $Res(f, g) = 0$ if and only if $f$ and $g$ have a common root, which is the same as saying that $f$ and $g$ have a common
divisor in $\mathbb{F}_q[x]$ of positive degree.

For  two polynomials $F(x,y),\, G(x,y)\in \mathbb{F}_q[x,y]$ of positive degree in $y$, the resultant $Res(F,G,y)$ of $F$ and $G$ with respect to $y$ is the resultant of $F$ and $G$ when considered as polynomials in the single variable $y$. 
In this case, $Res(F,G,y)\in \mathbb{F}_q[x] \cap \langle F,G\rangle$, where $\langle F,G\rangle$ is in the ideal generated by $F$ and $G$. Thus any pair $(a,b)$ with $F(a,b)=G(a,b)=0$ is such that $Res(F,G,y)(a)=0$. For more information on resultants and elimination theory, the reader can refer to \cite{CLO2007}.

\section{Differentially $4$ and $6$-uniform functions}\label{three}

In this section, we present new differentially $4$ and $6$-uniform permutation functions by modifying the function values of $x^{d}$ with $d=2^{4k}+2^{3k}+2^{2k}+2^{k}-1$ over a subfield of $\mathbb{F}_{2^{n}}$.

We firstly determine the number of solutions for a special equation, which will be used to prove Theorem \ref{th:cons}.

\begin{lemma}\label{lem:eq}
Let $k$ be a positive integer and $n=5k$. For any $b\in\mathbb{F}_{2^k}^{*}$, the equation
\begin{eqnarray}\label{eq:spe-eq}
(x+1)^{2^{4k}+2^{3k}+2^{2k}+2^{k}-1} + x^{2^{4k}+2^{3k}+2^{2k}+2^{k}-1} = b
\end{eqnarray}
has no solution in $\mathbb{F}_{2^n}\setminus \mathbb{F}_{2^k}$.
\end{lemma}

\begin{proof}
Assume that for a given $b\in\mathbb{F}_{2^k}^{*}$, Equation \eqref{eq:spe-eq}
has at least a solution $x$ in $\mathbb{F}_{2^n}\setminus \mathbb{F}_{2^k}$. 
Let $y=x^{2^{k}}$, $z=y^{2^{k}}$, $u=z^{2^{k}}$, $v=u^{2^k}$.
Then $x=v^{2^k}$ and for a given $b\in\mathbb{F}_{2^k}^{*}$, we study a system of five equations in $x$, $y$, $z$, $u$, $v$ obtained by
considering \eqref{eq:spe-eq}, \eqref{eq:spe-eq}$^{2^k}$, \eqref{eq:spe-eq}$^{2^{2k}}$, \eqref{eq:spe-eq}$^{2^{3k}}$, \eqref{eq:spe-eq}$^{2^{4k}}$.
Thus we have the equations system
\begin{subequations}  \label{eq:1.1}
\begin{empheq}[left=\empheqlbrace]{align}
  yzuv(x+1) + x(y+1)(z+1)(u+1)(v+1) + bx(x+1) =0,  \label{eq:1.1a} \\
  zuvx(y+1) + y(z+1)(u+1)(v+1)(x+1) + by(y+1) =0,  \label{eq:1.1b} \\
  uvxy(z+1) + z(u+1)(v+1)(x+1)(y+1) + bz(z+1) =0,  \label{eq:1.1c} \\
  vxyz(u+1) + u(v+1)(x+1)(y+1)(z+1) + bu(u+1) =0,  \label{eq:1.1d} \\
  xyzu(v+1) + v(x+1)(y+1)(z+1)(u+1) + bv(v+1) =0.  \label{eq:1.1e}
\end{empheq}
\end{subequations}
Computing the resultant of \eqref{eq:1.1a} and \eqref{eq:1.1d}, \eqref{eq:1.1b} and \eqref{eq:1.1d}, \eqref{eq:1.1c} and \eqref{eq:1.1d}, \eqref{eq:1.1e} and \eqref{eq:1.1d} with respect to $v$, we obtain
\begin{subequations}  \label{eq:1.2}
\begin{empheq}[left=\empheqlbrace]{align}
  &(x + u)^2(bxyz + bxyu + bxzu + bxu + y^2z^2 + y^2z + yz^2   \notag \\
  &  + byzu + byz + yz) = 0,  \label{eq:1.2b} \\
  &(y + u)^2(x^2z^2 + x^2z + bxyz + bxyu + xz^2 + bxzu + bxz   \notag \\
  &  + xz + byzu + byu) = 0,   \label{eq:1.2c} \\
  &(z + u)^2(x^2y^2 + x^2y + xy^2 + bxyz + bxyu + bxy + xy   \notag \\
  &  + bxzu + byzu + bzu) = 0,  \label{eq:1.2d}  \\
  & u(u + 1)(xy + xz + yz + x + y + z + b + 1)^2  \notag \\
  & \cdot (xyz + xyu + xzu + yzu + xu + yu + zu + bu^2 + bu + u)=0.  \label{eq:1.2a}
\end{empheq}
\end{subequations}
It is obvious that $u\neq0$ and $u\neq1$ in \eqref{eq:1.2a} since $x\in\mathbb{F}_{2^n}\setminus \mathbb{F}_{2^k}$. Moreover, we have $x, y, z\neq u$. Otherwise, we deduce $x\in\mathbb{F}_{2^k}$, which is a contradiction. Computing the resultant of \eqref{eq:1.2b}$/(x + u)^2$ and \eqref{eq:1.2a}$/u(u + 1)$, \eqref{eq:1.2c}$/(y + u)^2$ and \eqref{eq:1.2a}$/u(u + 1)$ , \eqref{eq:1.2d}$/(z + u)^2$ and \eqref{eq:1.2a}$/u(u + 1)$ with respect to $u$, we derive
\begin{subequations}  \label{eq:1.3}
\begin{empheq}[left=\empheqlbrace]{align}
  & bxyz(x+1)(y+1)(z+1)(y + z + b + 1)^2  \notag \\
  & \cdot (xy + xz + yz +  x + y + z + b + 1)^2 = 0,  \label{eq:1.3a} \\
  & bxyz(x+1)(y+1)(z+1)(x + z + b + 1)^2  \notag \\
  & \cdot (xy + xz + yz +  x + y + z + b + 1)^2 = 0,  \label{eq:1.3b} \\
  & bxyz(x+1)(y+1)(z+1)(x + y + b + 1)^2  \notag \\
  & \cdot (xy + xz + yz +  x + y + z + b + 1)^2 = 0.   \label{eq:1.3c}
\end{empheq}
\end{subequations}
Assume $x\in\mathbb{F}_{2^n}\setminus \mathbb{F}_{2^k}$, we have $x,y,z\neq0,1$. Suppose that $y + z + b + 1=0$ in (\ref{eq:1.3a}). Then $y + z \in\mathbb{F}_{2^k}$ since $b\in\mathbb{F}_{2^k}$. Raising $y + z$ to the $2^{3k}$-th power, we have $y + z = z+u$. Hence $y=u$, which means $x\in\mathbb{F}_{2^{2k}}$. This yields $x\in\mathbb{F}_{2^{2k}}\cap \mathbb{F}_{2^{5k}}=\mathbb{F}_{2^k}$, which is a contradiction. Similarly, we obtain $x + z + b + 1\neq0$ and  $x + y + b + 1\neq0$. Thus from the equation system \eqref{eq:1.3} we have
\begin{eqnarray}\label{eq:1.3-1}
 xy + xz + yz +  x + y + z + b + 1=0.
\end{eqnarray}
For $b\in\mathbb{F}_{2^k}^{*}$, raising both sides of \eqref{eq:1.3-1} to the $2^{k}$-th power, we obtain
\begin{eqnarray}\label{eq:1.4}
yz + yu + zu +  y + z + u + b + 1 = 0.
\end{eqnarray}
Computing the resultant of \eqref{eq:1.3-1} and \eqref{eq:1.4} with respect to $z$, we derive
\begin{eqnarray}  \label{eq:1.5}
  (x + u)(y^2 + y + b) = 0.
\end{eqnarray}
Recall that $x\neq u$. If $y^2 + y + b = 0$, then we have $y^2 + y \in\mathbb{F}_{2^k}$ since $b\in\mathbb{F}_{2^k}$. This yields $y^2 + y =y^{2^{k+1}}+y^{2^k}$, and we obtain
$y+y^{2^k} \in\mathbb{F}_{2}\subset\mathbb{F}_{2^k}$. We further derive $y+y^{2^k} = y^{2^k}+y^{2^{2k}}$, which means  $y\in\mathbb{F}_{2^{2k}}$. Hence $y\in\mathbb{F}_{2^{2k}}\cap \mathbb{F}_{2^{5k}}=\mathbb{F}_{2^k}$, which is  a contradiction. Thus \eqref{eq:1.5} has no solution.
Therefore we prove that \eqref{eq:spe-eq} has no solution in $\mathbb{F}_{2^n}\setminus \mathbb{F}_{2^k}$ for any $b\in\mathbb{F}_{2^k}^{*}$.
\end{proof}

\begin{theorem}\label{th:cons}
Let $k$ be an odd integer and $n=5k$. Define the function $ f(x)= g(x) + (g(x)+x^{d})\big(x+x^{2^{k}}\big)^{2^{n}-1}$ over $\mathbb{F}_{2^n}$,
where $d=2^{4k}+2^{3k}+2^{2k}+2^{k}-1$. Then we have the following:

\emph{(i)} if $g(x)$ is an APN permutation or a differentially $4$-uniform permutation function over $\mathbb{F}_{2^k}$, then $f(x)$ is a differentially $4$-uniform permutation function over $\mathbb{F}_{2^n}$;

\emph{(ii)} if $g(x)$ is a differentially $6$-uniform permutation function over $\mathbb{F}_{2^k}$, then $f(x)$ is also a differentially $6$-uniform permutation function over $\mathbb{F}_{2^n}$.
\end{theorem}

\begin{proof}
Firstly, we study the permutation property of $f(x)$. If $k$ is odd, then $\gcd(d, 2^n-1)=1$, which implies that $x^d$ is a permutation function over $\mathbb{F}_{2^n}$. For odd $k$, we have $x^d=x^3$ over $\mathbb{F}_{2^k}$ and $\gcd(3, 2^k-1)=1$. Thus $x^d$ is a permutation function over $\mathbb{F}_{2^k}$ and the image set $Im(x^d)$ of $x^d$ over $\mathbb{F}_{2^k}$ equals $\mathbb{F}_{2^k}$. Therefore, $x^d$ is a permutation function over  the complement set $\mathbb{F}_{2^n}\setminus \mathbb{F}_{2^{k}}$. Since $g(x)$ is a permutation function over $\mathbb{F}_{2^k}$, we obtain that $f(x)$ is a permutation function over $\mathbb{F}_{2^n}$.

Secondly, we calculate the differential uniformity of $f(x)$.
According to \cite[Proposition 3.1]{CM2019}, if $g(x)$ is an APN permutation, then we obtain that $f(x)$ is a differentially $4$-uniform permutation function over $\mathbb{F}_{2^n}$. For the cases (i) and (ii) when $g(x)$ is a differentially $4$ and $6$-uniform permutation function over $\mathbb{F}_{2^k}$, according to \cite[Theorem 3.1]{CM2019}, we have $\delta_{f}\leq \delta_{g}$, Thus it is needed to prove only that we have cases with $\delta_g$ solutions. 
In the following, we prove that there exists a pair $a, b \in\mathbb{F}_{2^n}$ and $a\neq 0$ such that the equation 
\begin{eqnarray}\label{eq:yuan}
f(x+a)+f(x)& = &
g(x+a)+[g(x+a)+(x+a)^{d}]((x+a)+(x+a)^{2^{k}})^{2^n-1}  \nonumber \\
&& + g(x) + (g(x)+x^{d})\big(x+x^{2^{k}}\big)^{2^{n}-1} =b
\end{eqnarray}
has exactly $\delta_{g}$ solutions.
Next we discuss separately the number of solutions of \eqref{eq:yuan} for the two distinct cases $a\in \mathbb{F}_{2^{k}}^{*}$ and $a\notin \mathbb{F}_{2^{k}}^{*}$.

(I) Suppose that $a\in \mathbb{F}_{2^{k}}^{*}$. If $x\in \mathbb{F}_{2^{k}}$ and $x+a\in \mathbb{F}_{2^{k}}$, then it follows from \eqref{eq:yuan} that
\begin{eqnarray*}
g(x+a)+g(x)=b.
\end{eqnarray*}
The equation has exactly $\delta_{g}$ solutions for some $b\in \mathbb{F}_{2^{k}}$ since $g(x)$ is a differentially $\delta_{g}$-uniform function over $\mathbb{F}_{2^k}$ or no solution for $b\in \mathbb{F}_{2^n}\setminus\mathbb{F}_{2^k}$.

If $x\in \mathbb{F}_{2^n}\setminus\mathbb{F}_{2^k}$ and $x+a\in \mathbb{F}_{2^n}\setminus\mathbb{F}_{2^k}$, then from \eqref{eq:yuan} we have
\begin{eqnarray}\label{eq:d}
(x+a)^{d} + x^{d}=b.
\end{eqnarray}
When $b\in \mathbb{F}_{2^{k}}$, substituting $x$ with $ay$ and dividing by $a^d$, we deduce $(y+1)^{d} + y^{d} = \frac{b}{a^d}$, where $\frac{b}{a^d}\in \mathbb{F}_{2^{k}}$. By Lemma \ref{lem:eq}, this equation has no solution $y\in \mathbb{F}_{2^n}\setminus\mathbb{F}_{2^k}$. Thus \eqref{eq:d} has no solution $x\in \mathbb{F}_{2^n}\setminus\mathbb{F}_{2^k}$.
When $b\in \mathbb{F}_{2^n}\setminus\mathbb{F}_{2^k}$,  \eqref{eq:d} has exactly two solutions for some $b\in \mathbb{F}_{2^n}\setminus\mathbb{F}_{2^k}$ since $x^d$ is an APN function.

Thus in the case $a\in \mathbb{F}_{2^{k}}^{*}$, \eqref{eq:yuan} has exactly $\delta_{g}$ solutions if $g(x)$ is a differentially $4$ and $6$-uniform function for some $b\in \mathbb{F}_{2^n}$.

(II) Suppose that $a\in \mathbb{F}_{2^n}\setminus\mathbb{F}_{2^k}$.
If $x\in \mathbb{F}_{2^n}\setminus\mathbb{F}_{2^k}$ and $x+a\in \mathbb{F}_{2^n}\setminus\mathbb{F}_{2^k}$, then we obtain that
\begin{eqnarray*}
(x+a)^{d} + x^{d}=b
\end{eqnarray*}
has exactly two solutions for some $b\in \mathbb{F}_{2^n}$ since $x^d$ is an APN function.

If $x\in \mathbb{F}_{2^n}\setminus\mathbb{F}_{2^k}$ and $x+a\in \mathbb{F}_{2^k}$, then we have
\begin{eqnarray}\label{eq:d1.5}
g(x+a) + x^{d}=b.
\end{eqnarray}
Raising both sides of \eqref{eq:d1.5} to the $2^{k}$-th power and then adding \eqref{eq:d1.5}, we derive
\begin{eqnarray}\label{eq:d1.6}
x^{d\cdot 2^k} + x^{d}=b + b^{2^k}.
\end{eqnarray}
Let $z=x^{d}$. Since we have shown that $x^{d}$ is a permutation over $x\in \mathbb{F}_{2^n}\setminus\mathbb{F}_{2^k}$, we obtain one $z$ by the $x$. Then \eqref{eq:d1.6} can be written as
\begin{eqnarray}\label{eq:d1.7}
z^{2^k} + z = b + b^{2^k}.
\end{eqnarray}
The solutions of \eqref{eq:d1.7} are in the coset $b+\mathbb{F}_{2^k}$. Since $x+a\in \mathbb{F}_{2^k}$, this yields $x\in a+\mathbb{F}_{2^k}$. Thus it suffices to compute the number of solutions $x\in (a+\mathbb{F}_{2^k})^d \cap (b+\mathbb{F}_{2^k})$, where $(a+\mathbb{F}_{2^k})^d = \{ x^d : x\in a+\mathbb{F}_{2^k} \}$. Assume that
$| (a+\mathbb{F}_{2^k})^d \cap (b+\mathbb{F}_{2^k})| \geq2$. Then there exist $u_{1}, u_{2}, v_{1}, v_{2}\in \mathbb{F}_{2^k}$ such that $(a+v_{1})^d = b+u_{1}$, $(a+v_{2})^d = b+u_{2}$ and $u_{1}\neq u_{2}$, $v_{1}\neq v_{2}$. Hence $$(a+v_{1})^d +(a+v_{2})^d = u_{1}+u_{2}.$$
Let us denote $x_{1}=a+v_{1}$, $a_{1}=v_{1}+v_{2}$. Then we obtain $x_{1}\in \mathbb{F}_{2^n}\setminus\mathbb{F}_{2^k}$ and
$$x_{1}^d +(x_{1}+a_{1})^d = u_{1}+u_{2}.$$
Since $a_{1}\in \mathbb{F}_{2^k}$ and by Lemma \ref{lem:eq}, this equation has no solution in $\mathbb{F}_{2^n}\setminus\mathbb{F}_{2^k}$, which is a contradiction. Thus $| (a+\mathbb{F}_{2^k})^d \cap (b+\mathbb{F}_{2^k})| \leq1$, which means that \eqref{eq:yuan} has at most one solution. Similarly, we can prove that \eqref{eq:yuan} has at most one solution in the case $x\in \mathbb{F}_{2^k}$ and $x+a\in \mathbb{F}_{2^n}\setminus\mathbb{F}_{2^k}$. Therefore for the case $a\in \mathbb{F}_{2^n}\setminus\mathbb{F}_{2^k}$, we obtain that \eqref{eq:yuan} has at most $4=\delta_{x^{d}}+2$ solutions for any $b\in \mathbb{F}_{2^n}$.

From the above discussion, if $g(x)$ is a differentially $4$-uniform function, then \eqref{eq:yuan} has exactly $4$ solutions for some $a, b \in\mathbb{F}_{2^n}$ and $a\neq 0$. Thus $f(x)$ is a differentially $4$-uniform function.
If $g(x)$ is a differentially $6$-uniform function, then \eqref{eq:yuan} has exactly $6$ solutions for some $a, b \in\mathbb{F}_{2^n}$ and $a\neq 0$. Therefore $f(x)$ is a differentially $6$-uniform function.
\end{proof}

\begin{remark}
In Theorem \ref{th:cons}, if $k$ is even, then $x^d$ is an APN non-permutation function \cite{Hans2001} . Thus $f(x)$ is not a permutation over $\mathbb{F}_{2^n}$, but the differential uniformity of $f(x)$ is invariant.
If $k$ is odd and $g(x)$ is not a permutation over $\mathbb{F}_{2^k}$, then $f(x)$ is not a permutation function, but the differential uniformity of $f(x)$ is invariant.  Therefore in Theorem \ref{th:cons}, we always assume that $k$ is odd and $g(x)$ is a permutation function over $\mathbb{F}_{2^k}$.
\end{remark}

\begin{corollary}\label{cor:power}
Let $k$ be an odd integer, $m$ be an integer, $n=5k$ and $\gcd(k, m)=1$. Define the function $g(x)=\mathcal{L}_{1} \circ x^{2^m-1}\circ \mathcal{L}_{2}$, where $\mathcal{L}_{1}, ~\mathcal{L}_{2}$ are affine permutations over $\mathbb{F}_{2^k}$. Define the function $ f(x)= g(x) + (g(x)+x^{d})\big(x+x^{2^{k}}\big)^{2^{n}-1}$ over $\mathbb{F}_{2^n}$,
where $d=2^{4k}+2^{3k}+2^{2k}+2^{k}-1$. Then the differential uniformity of $f(x)$ satisfies the following:

\emph{(i)} if $x^{2^m-1}$ is an APN or a differentially $4$-uniform function, then $f(x)$ is a differentially $4$-uniform permutation function;

\emph{(ii)} if $x^{2^m-1}$ is a differentially $6$-uniform function, then $f(x)$ is a differentially $6$-uniform permutation function.
\end{corollary}

For $m=k-1, \frac{k+1}{2}, 2$, the power functions $x^{2^m-1}$ are APN functions over $\mathbb{F}_{2^k}$. Moreover, the function $x^{2^{\frac{k+1}{2}}-1}$ is the inverse of a quadratic function. For more details the readers can refer to \cite{BCC2011}. Based on Corollary \ref{cor:power}, we directly obtain the following conclusion.

\begin{corollary}\label{cor:diff4}
Let $k$ be an odd integer and $n=5k$. Define the APN functions $g(x)=\mathcal{L}_{1} \circ x^{2^m-1}\circ \mathcal{L}_{2}$ with $m=k-1, \frac{k+1}{2}, 2$, where $\mathcal{L}_{1}, ~\mathcal{L}_{2}$ are affine permutations over $\mathbb{F}_{2^k}$.  Then
the functions $f(x)$ are differentially $4$-uniform permutations over $\mathbb{F}_{2^n}$. 
\end{corollary}

In the following, we discuss the algebraic degree and the nonlinearity of the function $f(x)$.

In \cite{CM2019} Calderini gave a necessary and sufficient condition for having $deg(f)=n-1$ if $f$ is a permutation in Lemma 4.3.
Furthermore, in \cite[Remark 4.3]{CM2019} the author can generalize Lemma 4.3 also for the case $deg(f)=n-t$, using monomials $x^{s}$ with $\omega_{2}(s)=t$.
Therefore we introduce the generalized result in the following.

\begin{lemma} \emph{(\cite{CM2019})} ~\label{lem:deg}
Let $f$ be a function defined over $\mathbb{F}_{2^n}$. Then $f$ in its polynomial representation
has a term of algebraic degree $n-t$ if and only if there exists a monomial $x^{s}$ with $\omega_{2}(s)=t$ such
that $\sum\limits_{x\in \mathbb{F}_{2^n}}f(x)x^{s}\neq0$, where $\omega_{2}(s)$ is the $2$-weight of the exponent $s$.
\end{lemma}

\begin{proposition}~\label{prop:deg}
Let $k>3$ be an odd integer and $n=5k$. Define the function $g(x)=\mathcal{L}_{1} \circ x^{2^m-1}\circ \mathcal{L}_{2}$, where $\mathcal{L}_{1}, ~\mathcal{L}_{2}$ are affine permutations over $\mathbb{F}_{2^k}$. Define the function $ f(x)= g(x) + (g(x)+x^{d})\big(x+x^{2^{k}}\big)^{2^{n}-1}$ over $\mathbb{F}_{2^n}$,
where $d=2^{4k}+2^{3k}+2^{2k}+2^{k}-1$.  Then the algebraic degree of $f(x)$ satisfies the following:

\emph{(i)} if $m=k-1$, then $deg(f)=n-1$;

\emph{(ii)} if $m=\frac{k+1}{2}$, then $\frac{9k+1}{2}\leq deg(f)\leq n-1$;

\emph{(iii)} if $m=2$, and there exist the affine permutations $\mathcal{L}_{1}, ~\mathcal{L}_{2}$ over $\mathbb{F}_{2^k}$
such that $deg(g+x^3)=2$, 
then $4k+2\leq deg(f)\leq n-1$.
\end{proposition}

\begin{proof}
Since $g(x)$ is affine equivalent to the function $x^{2^m-1}$ over $\mathbb{F}_{2^k}$, this yields $deg(g)=m\geq2$ for $k>3$. Since $deg(x^3)=2$, we have $deg(g+x^3)=m$. According to Lemma~\ref{lem:deg}, there exists a monomial $\varphi(x)=x^{s}\in\mathbb{F}_{2^k}[x]$ with $\omega_{2}(s)=k-m\leq k-2$ such that
$$\sum_{x\in \mathbb{F}_{2^k}}(g(x)+x^3)\varphi(x)\neq 0. $$
Thus we obtain
\begin{eqnarray}\label{eq:deg}
\sum_{x\in \mathbb{F}_{2^n}}f(x)\varphi(x) &=& \sum_{x\in \mathbb{F}_{2^k}}g(x)\varphi(x) + \sum_{x\in \mathbb{F}_{2^n}\setminus\mathbb{F}_{2^k}}x^d\varphi(x)  \nonumber \\
                                             &=& \sum_{x\in \mathbb{F}_{2^k}}g(x)\varphi(x) + \sum_{x\in \mathbb{F}_{2^n}}x^d\varphi(x) + \sum_{x\in \mathbb{F}_{2^k}}x^d\varphi(x) \nonumber \\
                                             &=& \sum_{x\in \mathbb{F}_{2^k}}g(x)\varphi(x) + \sum_{x\in \mathbb{F}_{2^n}}x^d\varphi(x) + \sum_{x\in \mathbb{F}_{2^k}}x^3\varphi(x) \nonumber \\
                                             &=& \sum_{x\in \mathbb{F}_{2^k}}(g(x)-x^3)\varphi(x) + \sum_{x\in \mathbb{F}_{2^n}}x^d\varphi(x).
\end{eqnarray}
Since for any function $h(x)\in\mathbb{F}_{2^n}[x]$, we have $\sum\limits_{x\in \mathbb{F}_{2^n}}h(x)\neq0$ if and only if $deg(h)=n$.
Since the algebraic degree of $x^d$ over $\mathbb{F}_{2^n}$ equals $deg(x^d)=k+3$, we deduce $deg(x^d\varphi(x))<n-1$. Thus we obtain $\sum\limits_{x\in \mathbb{F}_{2^n}}x^d\varphi(x)=0$.
It follows from \eqref{eq:deg} that
\begin{eqnarray*}
\sum_{x\in \mathbb{F}_{2^n}}f(x)\varphi(x) = \sum_{x\in \mathbb{F}_{2^k}}(g(x)-x^3)\varphi(x) \neq0.
\end{eqnarray*}
This means that there is a term of algebraic degree $n-(k-m)=4k+m$ in the polynomial representation of $f(x)$. Then the algebraic degree of $f(x)$ is $4k+m\leq deg(f)\leq n-1$ since $f(x)$ is a permutation function over $\mathbb{F}_{2^n}$. For (i), if $m=k-1$, then we have $deg(f)=n-1$. For (ii) and (iii), we have $\frac{9k+1}{2}\leq deg(f)\leq n-1$ and $4k+2\leq deg(f)\leq n-1$.

\end{proof}

\begin{remark} For (ii) and (iii), we think that the algebraic degrees equal the lower bound values by MAGMA.

For $k=1, 3$, the algebraic degree of $f(x)$ defined in Proposition \ref{prop:deg} is given as follows which has been confirmed by MAGMA.

(i) For the case of $m=k-1$, the algebraic degree of $f(x)$ defined in Proposition \ref{prop:deg} is $deg(f)=4$ when $k=1$ or $deg(f)=14$ when $k=3$.

(ii) For the case of $m=\frac{k+1}{2}$, we obtain the algebraic degree $deg(f)=4$ when $k=1$ or $deg(f)=14$ when $k=3$.

(iii) For the case of $m=2$, we obtain the algebraic degree $deg(f)=4$ when $k=1$ or $deg(f)=14$ when $k=3$.
\end{remark}

In \cite{Carlet2018} Carlet gave a lower bound on the nonlinearity of any APN power function. Using this result, we give a lower bound on the nonlinearity of $f(x)$ over $\mathbb{F}_{2^n}$.

\begin{lemma}\emph{(\cite{Carlet2018})}\label{lem:NL}
Let $f(x)$ be any APN power function. Then we have ~$\mathcal{NL}(f)\geq 2^{n-1}-2^{\frac{3n-3}{4}}$ if $n$ is odd and $\mathcal{NL}(f)\geq 2^{n-1}-2^{\frac{3n-2}{4}}$ if $n$ is even.
\end{lemma}

\begin{proposition}~\label{th:NL4}
Let $k$ be a positive integer and $n=5k$. Let the function $f(x)= g(x) + (g(x)+x^{d})\big(x+x^{2^{k}}\big)^{2^{n}-1}$ over $\mathbb{F}_{2^n}$,
where $d=2^{4k}+2^{3k}+2^{2k}+2^{k}-1$. Then the nonlinearity of the function $f(x)$ is
\begin{displaymath}\mathcal{NL}(f)\geq \left\{
                               \begin{array}{ll}
                                2^{n-1}-2^{\frac{3n-3}{4}}- 2^{\frac{k-1}{2}}-2^{k-1}, &  \textrm{~$k$ odd},\\
                                2^{n-1}-2^{\frac{3n-2}{4}}- 2^{\frac{k}{2}}-2^{k-1}, &  \textrm{~$k$ even}.
                               \end{array}\right.
\end{displaymath}
\end{proposition}

\begin{proof}
According to the Definition~\ref{def:walsh} of Walsh transform, we have
\begin{eqnarray*}
f^{\mathcal{W}}(u, v)&=&\sum_{x\in \mathbb{F}_{2^n}}(-1)^{{\rm Tr}(ux+vf(x))}  \\
          &=&\sum_{x\in \mathbb{F}_{2^n}}(-1)^{{\rm Tr}\{ux+v[g(x) + (x^{d}+g(x))\big(x+x^{2^{k}}\big)^{2^{n}-1}]\}}   \\
          &=&\sum_{x\in \mathbb{F}_{2^n}\setminus  \mathbb{F}_{2^{k}}}(-1)^{{\rm Tr}(ux+vx^{d})} + \sum_{x\in  \mathbb{F}_{2^{k}}}(-1)^{{\rm Tr}[ux+vg(x)]}  \\
          &=&\sum_{x\in \mathbb{F}_{2^n}}(-1)^{{\rm Tr}(ux+vx^{d})} - \sum_{x\in  \mathbb{F}_{2^{k}}}(-1)^{{\rm Tr}(ux+vx^{d})}
          + \sum_{x\in  \mathbb{F}_{2^{k}}}(-1)^{{\rm Tr}[ux+vg(x)]} .
\end{eqnarray*}
Since $x^d=x^3$ over $\mathbb{F}_{2^k}$, we have
\begin{displaymath}
\sum_{x\in \mathbb{F}_{2^{k}}}(-1)^{{\rm Tr}(ux+vx^{d})} = \sum_{x\in \mathbb{F}_{2^{k}}}(-1)^{{\rm Tr}(ux+vx^{3})} \leq \left\{
                               \begin{array}{ll}
                                2^{\frac{k+1}{2}}, &  \textrm{~$k$ odd},\\
                                2^{\frac{k+2}{2}}, &  \textrm{~$k$ even}.
                               \end{array}\right.
\end{displaymath}
According to the conclusion of the inequality, we derive
\begin{eqnarray*}
|f^{\mathcal{W}}(u, v)| &\leq & \bigg|\sum_{x\in \mathbb{F}_{2^n}}(-1)^{{\rm Tr}(ux+vx^{d})}\bigg|+ \bigg|\sum_{x\in  \mathbb{F}_{2^{k}}}(-1)^{{\rm Tr}(ux+vx^{3})} \bigg|
+ \bigg|\sum_{x\in  \mathbb{F}_{2^{k}}}(-1)^{{\rm Tr}[ux+vg(x)]}\bigg| \nonumber \\
& \leq & \bigg|\sum_{x\in \mathbb{F}_{2^n}}(-1)^{{\rm Tr}(ux+vx^{d})}\bigg| + \bigg|\sum_{x\in \mathbb{F}_{2^k}}(-1)^{{\rm Tr}(ux+vx^{3})}\bigg| + |\mathbb{F}_{2^{k}}| .  
\end{eqnarray*}
We further have
$$\mathcal{NL}(f)=2^{n-1}- \frac{1}{2}\max \big| f^{\mathcal{W}}(u, v)\big| \geq 2^{n-1}- \frac{1}{2} \max \big|  (x^{d})^{\mathcal{W}}(u, v)\big| - \frac{1}{2}\max \big|  (x^{3})^{\mathcal{W}}(u, v)\big| - 2^{k-1}.$$
In the above equation, we have $2^{n-1}- \frac{1}{2} \big| (x^{d})^{\mathcal{W}}(u, v)\big| \geq \mathcal{NL}(x^{d})$. By Lemma \ref{lem:NL}, we obtain
$\mathcal{NL}(f)\geq 2^{n-1}-2^{\frac{3n-3}{4}}- 2^{\frac{k-1}{2}}-2^{k-1}$ when $k$ is odd, or $\mathcal{NL}(f)\geq 2^{n-1}-2^{\frac{3n-2}{4}}- 2^{\frac{k}{2}}-2^{k-1}$ when $k$ is even.
\end{proof}

In the following tables, we consider some examples of $f(x)$ from Corollary \ref{cor:power} over small fields and perform some computations using MAGMA. Here we always assume the affine permutation $\mathcal{L}_{2}=x$. Let $\{\omega_{0}(f), \omega_{2}(f),\cdots, \omega_{\delta}(f)\}$  denote the differential spectrum, $deg(f)$ denote algebraic degree, $\mathcal{NL}(f)$ denote nonlinearity, $\mathcal{LB}$ denote the lower bound value of nonlinearity in Proposition \ref{th:NL4}, $\mathcal{MAX}$ denote the best known value of nonlinearity of the differentially $4$-uniform functions over $\mathbb{F}_{2^{n}}$.
$\beta$ is the primitive element of $\mathbb{F}_{2^k}$, which is a subfield of $\mathbb{F}_{2^n}$.

\begin{table}[H]
\centering
\caption{Some examples of $f(x)$ with $m=2$ in Corollary \ref{cor:power} over $\mathbb{F}_{2^{10}}$.}
\begin{tabular}{cccccc}
\hline\noalign{\smallskip}
  $\mathcal{L}_{1}$   &       $\{ \omega_{0}(f),\omega_{2}(f),\omega_{4}(f) \}$      & $ deg(f)$    &  $\mathcal{NL}(f)$  & $\mathcal{LB}$    & $\mathcal{MAX}$       \\
\noalign{\smallskip}\hline\noalign{\smallskip}
 $x+1$                           &  $ \{ 523776, 523776, 0 \}$                       & $8$        &   $472$             &    $380$    &    $480 $   \\
   $x+\beta$                     &  $ \{ 525759, 519810, 1983 \}$                    & $8$        &   $468$             &    $380$    &    $480 $   \\
 $\beta x+\beta$                &  $ \{ 525261, 520806, 1485 \}$                    & $8$        &  $469$              &    $380$      &    $480$     \\
 $\beta^2 x^2+\beta$            &  $\{ 524319, 522690, 543 \}$                     & $ 8 $       &   $471$              &    $380$      &    $480$   \\
 $\beta^2x^2$                   &  $ \{ 525261, 520806, 1485 \}$                    & $ 8 $        &  $469$              &    $380$      &    $480$     \\
\noalign{\smallskip}\hline
\end{tabular}
\label{tab1}
\end{table}

\begin{table}[H]
\centering
\caption{Some examples of $f(x)$ with $m=1$ in Corollary \ref{cor:power} over $\mathbb{F}_{2^{10}}$.}
\begin{tabular}{cccccc}
\hline\noalign{\smallskip}
  $\mathcal{L}_{1}$   &       $\{ \omega_{0}(f),\omega_{2}(f),\omega_{4}(f) \}$      & $ deg(f)$    &  $\mathcal{NL}(f)$  & $\mathcal{LB}$    & $\mathcal{MAX}$        \\
\noalign{\smallskip}\hline\noalign{\smallskip}
  $x+\beta$                     &  $\{ 524769, 521790, 993 \} $                    & $ 9 $        &   $470$             &    $380$    &    $480 $   \\
 $\beta x^2+\beta$                   &  $\{ 525309, 520710, 1533 \}$              & $ 9 $        &  $469$              &    $380$      &    $480$     \\
\noalign{\smallskip}\hline
\end{tabular}
\label{tab2}
\end{table}

When $k\geq 3$, MAGMA can not perform the program of the differential spectrum of $f(x)$ because the computation was terminated prematurely. So we only list examples of $f(x)$ over $\mathbb{F}_{2^{10}}~(k=2)$ in the above two tables. In Corollary \ref{cor:power}, for $k=2$, $f(x)$ is not a permutation function over $\mathbb{F}_{2^{10}}$, but the differential uniformity of $f(x)$ is invariant. Furthermore, from Tables \ref{tab1} and \ref{tab2}, it is easy to see that all these functions have high algebraic degree, but they have different differential spectra and nonlinearities, which means that they are CCZ-inequivalent to each other over $\mathbb{F}_{2^{10}}$.

\section{Conclusion}\label{conclu}
In this paper, by modifying the Dobbertin APN function $x^{2^{4k}+2^{3k}+2^{2k}+2^{k}-1}$ in a subfield of $\mathbb{F}_{2^{n}}$, we constructed new classes of low differentially uniform permutation functions. Furthermore, we determined a lower bound of nonlinearity and computed the algebraic degree of functions with explicit form.

\end{document}